\documentclass[a4paper,12pt]{article}
\usepackage{amsmath,amssymb,amsfonts,amsthm}
\newtheorem{theorem}{Theorem}[section]

\usepackage{amsmath,amssymb,amsfonts,amsthm}
\usepackage[utf8]{inputenc}

\usepackage{graphicx}
\usepackage{enumerate}
\usepackage{amsmath,amssymb,amsfonts,amsthm}

\usepackage{t1enc}
\usepackage{color}

\title{Comment on \textit{Horizon area- Angular momentum inequality for a class of axially symmetric black holes}}
\author{Mar\'ia E. Gabach Cl\'ement}
\begin{document}
\maketitle
\abstract{We extend the results presented by Aceña \textit{et al} in the afore mentioned paper, \cite{Acena11}, to the case of axisymmetric, maximal initial data which are invariant under an inversion transformation.}
\section{Introduction}

In \cite{Acena11} it is proven that for a large class of axisymmetric, vacuum and maximal initial data, with a non-negative
cosmological constant, having a surface
$\Sigma=\{r=constant\}$ where the following local conditions are
satisfied
\begin{align}
H  &=0,\label{cond1} \\ 
\partial_rH &\geq0, \label{cond2}\\
 \partial_rq &=0,\label{cond3}
\end{align}
then the following inequality holds  
\begin{equation}\label{desigualdad}
8\pi |J|\leq A,
\end{equation}
where $A$ is the area and $J$ the angular momentum of $\Sigma$.

We refere the reader to \cite{Acena11} for the detailed description of the axisymmetric, maximal initial data.

In this note, we prove that that the above result also applies to more general surfaces $\Sigma=\{r-Rc(\theta)=0\}$ if one  assumes that $\Sigma$ is an isometry surface, together with condition \eqref{cond2} for its mean curvature. 

This result is especially relevant for the case where there are multiple black holes, since in that situation one would expect that the surface of isometry is not one with constant radius.

\section{Main Result}

Consider an axially symmetric, maximal initial data set for Einstein's equations having the following induced 3-metric on a spatial hipersurface $S$ 
\begin{equation}\label{metric}
h=e^ {\sigma}\left[e^ {2q}(dr^ 2+r^2d\theta^ 2)+r^2\sin^2\theta(d\phi+v_r dr
  +v_\theta d\theta)^ 2\right],
\end{equation}
where $\sigma, q, v_r$ and $v_\theta$ are regular functions of $r$ and
$\theta$. Moreover, the data have angular momentum $J$ (see \cite{Acena11} for information concerning the extrinsic curvature of the data).

Let the 3-metric $h$ be invariant under the inversion transformation $r\to R^2c^2(\theta)/r$, where $c(\theta)$ is a smooth function of $\theta$ and let $\Sigma$ be a 2-surface that is fixed by such transformation. Then the isometry implies the following conditions
\begin{equation}\label{conditions}
H|_\Sigma=0,\qquad\partial_rq|_\Sigma=0,\qquad v_r|_\Sigma=0.
\end{equation}

The area of $\Sigma$ is given by
\begin{equation}
A=\int_0^ {2\pi}\int_0^\pi e^ {\sigma+q} R^ 2\sin\theta|c|\sqrt{(c^ 2+c'^ 2)}d\theta d\phi,
\end{equation}
and one can easily check that
\begin{equation}
A\geq\int_0^ {2\pi}\int_0^\pi e^ {\sigma+q} R^ 2c^2\sin\theta d\theta d\phi.
\end{equation}
We define the function $\varsigma$ in terms of $\sigma$ in the following way
\begin{equation}
\varsigma:=\sigma+2\ln r
\end{equation}
and obtain
\begin{equation}\label{a22}
A\geq2\pi\int_0^\pi e^ {\varsigma+q} \sin\theta d\theta.
\end{equation}
Then, using Einstein's constraints and the maximality condition, we can bound the right hand side of \eqref{a22}  with the mass functional $\mathcal M$ introduced in \cite{Acena11}
\begin{equation}\label{aa1}
A\geq4\pi e^{\frac{\mathcal M}{8}}e^{\frac{F+G}{16\pi}},
\end{equation}
where $\mathcal M, F, G$ are avaluated at $\Sigma$.

Then we obtain the following theorem

\begin{theorem}
Consider axisymmetric, vacuum and maximal initial data, with a non-negative
cosmological constant as described above. Assume there exists a surface
$\Sigma$ which is invariant under an inversion transformation and such that
\begin{equation}\label{derh}
\partial_rH \geq0
\end{equation}
holds on $\Sigma$. Then we have 
\begin{equation}\label{desigualdad}
8\pi |J|\leq A
\end{equation}
where $A$ is the area and $J$ the angular momentum of $\Sigma$.
\end{theorem}
\begin{proof}

In \cite{Acena11} it was proven that if the first two conditions in \eqref{conditions}, and \eqref{derh} are valid, then the right hand side of \eqref{aa1} is bounded from below by $8\pi |J|$: since the first two conditions are automatically satisfied for an inversion-fixed surface, we obtain the desired result.
\end{proof}

\textbf{Acknowledgments}. We want to thank Sergio Dain and Andr\'es Aceña for useful discussions.

\end{document}